\documentclass[submission,copyright,creativecommons]{eptcs}
\providecommand{\event}{LSFA 2012} 

\usepackage{breakurl}             
\usepackage{amssymb}
\usepackage{amsmath}
\usepackage{mathrsfs}
\usepackage{stmaryrd}
\usepackage{extarrows}
\usepackage{amsthm}
\usepackage{xifthen} 
\usepackage{rotating}
\usepackage{enumerate}
\usepackage{rotating}

{\begin{list}%
       {-}%
       {\setlength{\itemsep}{0pt}
       \setlength{\parsep}{3pt}
       \setlength{\topsep}{3pt}
       \setlength{\partopsep}{0pt}
       \setlength{\leftmargin}{0.7em}
       \setlength{\labelwidth}{1em}
       \setlength{\labelsep}{0.3em}}}%
{\end{list}}

{\begin{list}%
       {-}%
       {\setlength{\itemsep}{0pt}
       \setlength{\parsep}{2pt}
       \setlength{\topsep}{2pt}
       \setlength{\partopsep}{0pt}
       \setlength{\leftmargin}{2em}
       \setlength{\labelwidth}{1em}
       \setlength{\labelsep}{0.3em}}}%
{\end{list}}

\newtheorem{theorem}{Theorem}[section]

\newtheorem{lemma}[theorem]{Lemma}

\newtheorem{definition}{Definition}[section]

\newcommand{\veN}{\vec{N}}

\newcommand{\bA}{\mathbf{A}}

\newcommand{\bF}{\mathbf{F}}

\newcommand{\bI}{\mathbf{I}}

\newcommand{\bT}{\mathbf{T}}

\newcommand{\ga}{\alpha}
\newcommand{\gb}{\beta}
\newcommand{\gc}{\gamma}
\newcommand{\gd}{\delta}

\newcommand{\gr}{\rho}
\newcommand{\gs}{\sigma}
\newcommand{\gt}{\tau}

\newcommand{\go}{\omega}

\newcommand{\bgp}{\mbox{\boldmath $\pi$}}

\newcommand{\cT}{\mathcal{T}}

\newcommand{\ssim}{\stackrel{\mathsf{s}}{\sim}} 

\newcommand{\nat}{\mathbb{N}}

\newcommand{\sut}{:} 

\newcommand{\labelto}[1]{\rightarrow_{#1}} 
\newcommand{\mslabelto}[1]{\twoheadrightarrow_{#1}} 

\newcommand{\gramm}{\mathrel{::=}} 
\newcommand{\ass}{\mathrel{:=}} 

\newcommand{\seq}[1]{\vec{#1}}

\newcommand{\at}{\!\centerdot} 
\newcommand{\ats}{\at\ldots\at} 
\newcommand{\aps}{\ap\cdots\ap} 

\newcommand{\sps}[3]{\bgp^{(#1,#2,#3)}} 
\newcommand{\spt}[1]{\bA^{(#1)}} 

\newcommand{\s}{\mathsf{s}} 
\newcommand{\eo}{\eta_1} 
\newcommand{\et}{\eta_2} 

\newcommand{\cdr}[1]{\mathsf{cdr}(#1)} 
\newcommand{\car}[1]{\mathsf{car}(#1)} 
\newcommand{\cddr}[2]{#1[#2)} 
\newcommand{\cadr}[2]{#1[#2]} 

\newcommand{\nil}{\mathsf{nil}} 

\newcommand{\hnf}[1]{\mathsf{Hnf}(#1)} 
\newcommand{\bdom}[2]{\mathsf{dom}(#1,#2)} 
\newcommand{\dom}[1]{\mathsf{dom}(#1)} 
\newcommand{\trun}[2]{#1 \!\upharpoonright {#2}} 
\newcommand{\virt}[1]{\mathsf{vir}(#1)} 
\newcommand{\simty}{\stackrel{\infty}{\sim}} 
\newcommand{\pexp}[2]{\mbox{\boldmath{$\langle$}} #1 \lVert #2 \mbox{\boldmath{$\rangle$}}} 
\newcommand{\Seq}{\mathsf{Seq}} 

\newcommand{\opeq}{\approx} 

\newcommand{\sub}[2]{\{#1/#2\}} 
\newcommand{\ap}{\star} 
\newcommand{\bd}{\mu} 

\newcommand{\FV}{\mathrm{FV}} 
\newcommand{\KTer}[1]{\Sigma^{\mathsf{#1}}} 
\newcommand{\lmu}{\ensuremath{\lambda\mu}} 

\newcommand{\leng}[1]{\sharp #1} 

\newcommand{\wei}[1]{\mathsf{w}(#1)} 
\newcommand{\bwei}[2]{\mathsf{w}(#1,#2)} 
\newcommand{\brea}[1]{\mathsf{b}(#1)} 
\newcommand{\bbrea}[2]{\mathsf{b}(#1,#2)} 

\newcommand{\tystk}[4]{%
\ifthenelse{\equal{#1}{s}\OR\equal{#1}{t}}{
	\ifthenelse{\equal{#1}{s}}{
		\ifthenelse{\isempty{#2}}{#3 \vdash #4}{{#2} {:} #3 \vdash #4}
		}{
		\ifthenelse{\isempty{#2}}{\vdash #3, #4}{\vdash {#2} {:} #3 \ {\mid} \ #4}
		}
	}{
	\ifthenelse{\isempty{#2}}{\vdash #4}{\vdash {#2} {\mid} #4}
	}
}
\newcommand{\ntystk}[4]{%
\ifthenelse{\equal{#1}{s}\OR\equal{#1}{t}}{
	\ifthenelse{\equal{#1}{s}}{
		\ifthenelse{\isempty{#2}}{#3 \nvdash #4}{{#2} {:} #3 \nvdash #4}
		}{
		\ifthenelse{\isempty{#2}}{\nvdash #3, #4}{\nvdash {#2} {:} #3 \ {\mid} \ #4}
		}
	}{
	\ifthenelse{\isempty{#2}}{\nvdash #4}{\nvdash {#2} {\mid} #4}
	}
}

\newcommand{\tylmu}[4]{
	\ifthenelse{\isempty{#3}}
	{
	\ifthenelse{\isempty{#4}}
		{#1 \vdash_{\lmu} {#2}}
		{#1 \vdash_{\lmu} {#2} \mid #4}
	}
	{
        \ifthenelse{\isempty{#4}}
        	{#1 \vdash_{\lmu} {#2} {:} #3}
        	{#1 \vdash_{\lmu} {#2} {:} #3 {\mid} #4}
	}
}

\input prooftree.sty

\title{The untyped stack calculus and B\"{o}hm's theorem}

\author{
Alberto Carraro
\institute{PPS, Universit\'{e} Paris Diderot, France}
\email{acarraro@pps.univ-paris-diderot.fr}
}

\def\titlerunning{The untyped stack calculus and B\"{o}hm's theorem}
\def\authorrunning{A. Carraro}
\begin{document}
\maketitle

\begin{abstract}
The stack calculus is a functional language in which is in a Curry-Howard correspondence with classical logic.
 It enjoys confluence but, as well as Parigot's $\lambda\mu$, does not admit the B\"{o}hm Theorem, typical of the $\lambda$-calculus.
 We present a simple extension of stack calculus which is for the stack calculus what Saurin's $\Lambda\mu$ is for $\lambda\mu$.
\end{abstract}

\section{Introduction}\label{sec:intro}

In \cite{Bohm68} Corrado B\"{o}hm proved the so-called \emph{B\"{o}hm's theorem}, a fundamental syntactical feature of the
 pure $\lambda$-calculus which states that if $M$ and $N$ are two distinct $\beta\eta$-normal terms, then for each pair of terms $P,Q$ there exists a
 context $C[\cdot]$ such that $C[M]$ is $\beta$-equivalent to $P$ and $C[N]$ is $\beta$-equivalent to $Q$. If moreover $M$ and $N$ are closed,
 the context $C[\cdot]$ can have shape $[\cdot]\seq L$, for a suitable sequence $\seq L = L_1,\ldots,L_k$ of $\lambda$-terms. The original issue motivating this result was the quest for solutions of systems of equations between $\lambda$-terms: given closed $\lambda$-terms
 $M_1,N_1,\ldots,M_n,N_n$, is there a $\lambda$-term $S$ such that $SM_1 =_\beta N_1 \wedge \cdots \wedge SM_n =_\beta N_n$ holds?
 The answer is trivial for $n = 1$ (just take $S = \lambda z.N_1$ for a fresh variable $z$) and B\"{o}hm's theorem gives a positive answer for $n = 2$ when
 $M_1,M_2$ are distinct $\beta\eta$-normal forms (apply the theorem to $M_1$ and $M_2$ and then set $S = \lambda f.f\seq L N_1 N_2$). The
 result has been then generalized (and this step is non-trivial) in \cite{BohmDPR79} to treat every finite family $M_1,\ldots,M_n$ of pairwise distinct
 $\beta\eta$-normal forms.

The notion of \emph{operational equivalence} has been a subject of many research works in the literature. Essentially one considers as ``equivalent'' two $\lambda$-terms $M$ and $N$ when for every possible context $C[\cdot]$ the head
 reduction process of $C[M]$ halts iff the head reduction process of $C[N]$ halts. It is common to think of two non-operationally equivalent terms
 as programs that can be distinguished one another by making them interact with all possible environments, observing termination of head reduction.
 From B\"{o}hm's theorem it follows that given two distinct $\beta\eta$-normal forms one can choose a term $P$ that admits a head normal
 form and a term $Q$ does not have a head normal form, having the guarantee that there exists a context $C[\cdot]$ such that
 $C[M]=_\beta P$ and $C[N]=_\beta Q$. In this sense $M$ and $N$ would be ``separated'' (or distinguished) by the context $C[\cdot]$ witnessing
 that $M$ and $N$ are not operationally equivalent. For this reason B\"{o}hm's theorem is also known as the \emph{separation theorem} and it is said to prove the \emph{separation property}
 for the untyped $\lambda$-calculus.
 
The separation property has consequences both on the semantical and on the syntactical side. For example it implies that $\beta\eta$-equivalence is
 the maximal non-trivial congruence on normalizable terms extending the $\beta$-equivalence, so that any model of the $\lambda$-calculus cannot
 identify two different $\beta\eta$-normal forms without being trivial. A possible reading of B\"{o}hm's theorem is that the $\lambda$-calculus is
 powerful enough to ``inspect'' itself and the syntax and the reduction rules fit each other well. Nonetheless B\'{o}hm's result gives an alternative characterization of operational equivalence for normalizable terms. The complete characterization
 operational equivalence (also for non-normalizable terms) was then achieved by Hyland \cite{Hyl76} and Wadsworth \cite{Wad76}.
 In general two $\lambda$-terms are operationally equivalent iff they have the same B\"{o}hm tree, up to possibly infinite $\eta$-expansion iff they
 have the same denotation in Scott's $D_\infty$ model (see also \cite{Bare84}). B\"{o}hm's proof shows how to produce the separating context with an
 algorithm whose inputs are $M$, $N$ ($\beta\eta$-normal terms) $P$ and $Q$ (arbitrary terms). Following this observation Huet \cite{Huet93} shows an ML implementation of B\"{o}hm's algorithm and poses the problem of formalizing
 a proof of B\"{o}hm's theorem for the purpose of mechanical checking. The combinatorial core of B\"{o}hm's algorithm is called the
 \emph{B\"{o}hm-out technique} and it is at the basis of the implementation, presented in \cite{Bohm94}, of the CUCH-machine, a $\lambda$-calculus
 interpreter introduced by B\"{o}hm and Gross in \cite{BohmGross66}. Various generalizations/extensions of B\"{o}hm's theorem have been studied. The $\lambda$-calculus has been immersed in other languages
 in order to obtain finer observations on the behaviour of $\lambda$-terms. Sangiorgi \cite{Sang94} considers the encoding of the $\lambda$-calculus in
 the $\pi$-calculus with the addition of a unary non-deterministic operator. In \cite{Dez96} and \cite{Dez99} Dezani et al. add a
 binary parallel operator and a non-deterministic choice. Manzonetto and Pagani \cite{ManPag11} give a proof of B\"{o}hm's theorem for a resource-sensitive
 extension of the $\lambda$-calculus.

Curry-Howard correspondence \cite{Howard80} was first stated as the isomorphism between natural deduction for minimal intuitionistic
 logic \cite{Prawitz65} and the simply typed $\lambda$-calculus. Later Griffin \cite{Griffin90} proposed that natural deduction for
 classical logic could be viewed as a type system for a $\lambda$-calculus with a certain control operator introduced by Felleisen \cite{Felleisen90}.
 Several other proposals have been made for a computational interpretation of classical logic, among which Parigot's $\lambda\mu$-calculus \cite{Parigot91}
 received a lot of attention. In \cite{David01} David and Py proved that the $\lambda\mu$-calculus does not satisfy the separation property.
 Saurin \cite{Saurin05} exhibited an extension of $\lambda\mu$-calculus, the $\Lambda\mu$-calculus, in which the separation property does hold.
 In \cite{CaEhSa12} the authors introduce the stack calculus, a finitary functional language in which the $\lambda\mu$-calculus can be faithfully
 translated. 

In the present paper we prove, using David and Py's \cite{David01} counterexample, that the separation property does not hold
 for the stack calculus, i.e., there are (extensionally) different normal forms which are operationally equivalent. We
 introduce the extended stack calculus which is a calculus that contains the stack calculus (as Saurin's $\Lambda\mu$ extends
 Parigot's $\lambda\mu$). We show that operational equivalence is maximally consistent, i.e. it cannot be properly extended to
 another consistent equational theory both in the stack calculus and in the extended stack calculus. We work out the details of
 a B\"{o}hm-out technique for the extended stack calculus. A nice feature of the extended stack calculus is that, having only 
 one binder (instead of two as in $\lambda\mu$), it admits a simpler proof of B\"{o}hm's theorem, which is similar to the one
 for the $\lambda$-calculus.

The treatment of B\"{o}hm's theorem deserves a prominent place in classical monographs on the $\lambda$-calculus (Hindley\textendash Seldin
 \cite{Hindley86}, Hankin \cite{Hankin95}, Barendregt \cite{Bare84}). Besides the applications of B\"{o}hm's theorem, there has always been interest
 around the proof itself and the algorithmic content of B\"{o}hm-out technique. In \cite{DGP08} Dezani et al. give a thorough account of B\"{o}hm's
 theorem, together with an overview of the impressive research activity which originated from it. To the best of our knowledge since Huet's
 challenge \cite{Huet93}, no mechanical proof of B\"{o}hm's theorem has been produced yet. Instead Aehlig and Joachimski \cite{Aehlig02} provide a
 different proof of B\"{o}hm's theorem that does not use the B\"{o}hm-out technique. In view of the interest in proofs of B\"{o}hm's
 theorem for various calculi, we believe useful to contribute in the present work with a direct proof of B\"{o}hm's
 theorem (i.e. with a B\"{o}hm-out technique) for the extended stack calculus, even if the mere separation
 result would follow by a suitable mutual translation with the $\Lambda\mu$-calculus.
 
\section{The untyped stack calculus}\label{sec:untyped_stack}

We report the presentation of stack calculus from \cite{CaEhSa12}. The language has three syntactic categories: \emph{terms} that are in functional position, \emph{stacks} that are in argument position and
represent streams of arguments, \emph{processes} that are terms applied to stacks.

The basis for the definition of the stack calculus language is a countably infinite set of \emph{stack variables}, ranged over by the initial small
 letters $\alpha,\beta,\gamma,\ldots$ of the greek alphabet. The language is then given by the following grammar:
\[
\begin{array}{llllr}
\text{stacks}    & \quad \pi,\varpi & \gramm & \nil    \mid \ga \mid \cdr \pi \mid M \at \pi & \qquad\qquad\qquad\qquad\qquad  \\
\text{terms}     & \quad M,N        & \gramm & \car\pi \mid \bd\ga.P &  \\
\text{processes} & \quad P,Q        & \gramm & M \ap \pi &  \\
\end{array}
\]
We use letters $E,E'$ to range over \emph{expressions} which are either stacks, terms or processes. The operator $\bd$ is a binder. 
An occurrence of a variable $\ga$ in an expression $E$ is \emph{bound} if it is under the scope of a $\bd\ga$; the set $\FV(E)$ of
 \emph{free variables} is made of those variables having a non-bound occurrence in $E$.

\noindent\emph{Stacks} represent lists of terms: $\nil$ is the empty stack. A stack $M_1\at \cdots \at M_k \at \nil$, stands for a finite
 list while a stack $M_1\at \cdots \at M_k \at \ga$ stands for a non-terminated list that can be further extended.
\emph{Terms} are entities that represent the ``active part'' of computations. A term $\bd\alpha.P$ is the \emph{$\bd$-abstraction} of $\alpha$ in $P$.\\
\emph{Processes} result from the \emph{application} $M\ap\pi$ of a term $M$ to a stack $\pi$. This application, unlike in the $\lambda$-calculus, has
 to be thought as \emph{exhaustive}: an application of a stack to a term is an evolving entity that does not have any outcome.\\

As usual, the calculus involves a substitution operator. By $E\sub{\pi}{\ga}$ we denote the substitution of the stack $\pi$ for all free occurrences
 of $\ga$ in $E$ (paying attention to avoid capture of free variables). A basic but useful fact about substitutions is the \emph{substitution lemma}
 (see \cite{CaEhSa12}): for all $E\in\KTer{e}$ and all $\pi,\varpi\in\KTer{s}$ with $\alpha\not\in\FV(\varpi)$ and $\ga \not\equiv \gb$ we have
 $E\sub{\pi}{\ga}\sub{\varpi}{\gb} \equiv E\sub{\varpi}{\gb}\sub{\pi\sub{\varpi}{\gb}}{\ga}$ (the symbol `$\equiv$' stands for syntactic equality).

The reduction rules characterizing the stack calculus are the following ones:
\[
\begin{array}{rl}
(\bd)          & \quad (\bd{\ga}.P)\ap\pi \labelto{\bd} P\sub{\pi}{\ga} \qquad\qquad\qquad\qquad\qquad\qquad\qquad\qquad\qquad \\
(\mathsf{car}) & \quad \car{M\at\pi}      \labelto{\mathsf{car}} M \\
(\mathsf{cdr}) & \quad \cdr{M\at\pi}      \labelto{\mathsf{cdr}} \pi
\end{array}
\]
Adding the following rules we obtain the \emph{extensional} stack calculus:
\[
\begin{array}{rlr}
(\eo) & \quad \bd{\ga}.M\ap\ga  \labelto{\eo} M   & \text{ if } \ga \not\in \FV(M) \qquad\qquad\qquad\qquad\qquad\qquad\\
(\et) & \quad \car\pi\at\cdr\pi \labelto{\et} \pi & 
\end{array}
\]

We simply write $\labelto{\s}$ for the contextual closure of the relation $(\labelto{\bd} \cup \labelto{\mathsf{car}} \cup \labelto{\mathsf{cdr}})$.
 Moreover we write $\labelto{\eta}$ for the contextual closure of the relation $(\labelto{\eo} \cup \labelto{\et})$ and finally we set
 $\labelto{\s\eta} = (\labelto{\s} \cup \labelto{\eta})$. We denote by $\mslabelto{\s}$ (resp. $\mslabelto{\s\eta}$) the reflexive and transitive closure
 of $\labelto{\s}$ (resp. $\labelto{\s\eta}$) and we denote by $=_{\s}$ (resp. $=_{\s\eta}$) the reflexive, symmetric, and transitive closure
 of $\labelto{\s}$ (resp. $\labelto{\s\eta}$).

An example of term is $\bI \ass \bd \ga.\car\ga \ap \cdr\ga$ (the symbol `$\ass$' stands for definitional equality). 
 For example $\bI\ap\bI\at\nil \labelto{\s} \bI\ap\nil \labelto{\s} \car\nil\ap\cdr\nil$ and the reduction does not proceed further.
 If $\go \ass \bd \ga.\car{\ga}\ap\ga$, then $\go\ap\go\at\nil \labelto{\s} \go\ap\go\at\nil$; this is an example of a non-normalizing process.
 The stack calculus enjoys confluence, even in its extensional version, as stated in the following theorem.

\begin{theorem}[\cite{CaEhSa12}]\label{thm:CR-bd-original}
The reductions $\labelto{\s}$ and $\labelto{\s\eta}$ are both Church-Rosser.
\end{theorem}

It seems natural to define a meta-language with constructions $\cddr{\pi}{n} \ass \mathsf{cdr}(\cdots\mathsf{cdr}(\pi)\cdots)$ ($n$ times)
 and $\cadr{\pi}{n} \ass \car{\cddr{\pi}{n}}$. It is easily checkable that every expression $E$ has a $\mslabelto{\mathsf{car},\mathsf{cdr}}$-normal
 form, that we will refer to as the \emph{canonical form} of $E$. If $\veN = N_1,\ldots,N_m$ is a sequence of terms, we write $\veN \at \cddr{\gc}{k}$ and $\veN \at \cddr{\nil}{k}$ for the obvious corresponding
 stacks. In an expression in canonical form the stacks have either shape $\veN \at \cddr{\gc}{k}$ or $\veN \at \cddr{\nil}{k}$ and the 
 non-abstraction terms have shape either $\cadr{\ga}{n}$ or $\cadr{\nil}{n}$.

\subsection{Operational equivalence and failure of B\"{o}hm's theorem}\label{subsec:outer-red}

We provide a notion of \emph{outer reduction} for the stack calculus, obtained by performing the contraction of outer-most redexes only.
 The one-step \emph{outer-reduction} on terms is given by the following rule:
 $$ M \labelto{o} \bd\ga.P\sub{\pi}{\gb} \qquad \text{ if } \bd\ga.(\bd \gb.P) \ap \pi \text{ is the canonical form of } M $$

Note that outer-reduction is \emph{deterministic} and we don't take any contextual closure for $\labelto{o}$; its reflexive and
 transitive closure is denoted by $\mslabelto{o}$. A term $M$ \emph{is in outer-normal form} (\emph{onf}, for short) if it is not
 $\labelto{o}$-reducible. It is straightforward to see that a term is in onf iff it has the form $\bd\ga.H\ap \veN\at\gt$, where the terms in the
 sequence $\veN$ are arbitrary, $H$ is either $\cadr{\gb}{n}$ (in which case we say that the onf is \emph{proper}) or $\cadr{\nil}{n}$
 (in which case we say that the onf is \emph{improper}) and $\gt$ is either $\cddr{\gc}{n}$ or $\cddr{\nil}{n}$; $H$ is the head of the onf and
 $\gt$ is the tail of the onf in question.
 A term \emph{has} a onf if it $\mslabelto{o}$-reduces to a term in onf. Of course the outer reduction strategy is complete for finding
 onf's of terms, i.e., if $M$ $\mslabelto{\s}$-reduces to a term $N'$ in onf, then $M$ reduces to some onf $N''$ via outer reduction.

A \emph{head context} is a context generated by the grammar $C[\cdot] \gramm [\cdot] \mid \bd\ga.C[\cdot]\ap\pi$.

\begin{definition}[Operational equivalence]\label{def:op-equiv}
Two terms $M,N$ are \emph{operationally equivalent}, notation $M \opeq N$, if for every head context $C[\cdot]$ we have that
 $C[M]$ has a proper onf iff $C[N]$ has a proper onf.
\end{definition}

Operational equivalence is a fairly common notion. In Definition \ref{def:op-equiv} we only quantify over head contexts, but it can be shown
 (as it is done for the $\lambda$-calculus \cite{Bare84} and the $\lambda\mu$-calculus \cite{David01}) that a quantification over \emph{all}
 contexts gives as equivalent definition.

%

We now define some important terms: $\bT\ass\bd\ga.\cadr{\ga}{0}\ap\cddr{\ga}{2}$, $\bF\ass\bd\ga.\cadr{\ga}{1}\ap\cddr{\ga}{2}$, and 
 $\Omega\ass\bd\gc.\go\ap\go\at\gc$, where $\go\ass\bd\ga.\cadr{\ga}{0}\ap\ga$. Clearly $\Omega$ is an example of term without onf
 and we have $\cadr{\nil}{0} \opeq \Omega \not\opeq \bT \not\opeq \bF$.

We conclude the first part of the paper showing that in the stack calculus there exist different $\labelto{\s\eta}$-normal forms which are
 operationally equivalent. This situation is in contrast with the $\lambda$-calculus: the original B\"{o}hm's theorem \cite{Bohm68} implies that
 two different $\beta\eta$-normal $\lambda$-terms $M$ and $N$ are never operationally equivalent because there exists a context $C[\cdot]$ such that
 $C[M]$ has a head normal form and $C[N]$ does not have a head normal form. However an analogous situation occurs for the $\lambda\mu$-calculus:
 David and Py \cite{David01} exhibited two extentionally different normal $\lambda\mu$-terms which are operationally equivalent.
 In fact the counterexample given in the next theorem is obtained by translating David and Py's terms into the stack calculus.

\begin{theorem}\label{thm:wrapper-exists}
Let $U \ass \bd\gc.\cadr{\ga}{0}\ap\ga$ and
 \mbox{$W[\cdot]\ass \bd\ga.\cadr{\ga}{0}\ap (\bd\gb.\cadr{\ga}{0}\ap U\at[\cdot]\at \ga)\at U \at \ga$}.
 Then for all terms $M,N$ we have \mbox{$W[M] \opeq W[N]$}.
\end{theorem}

\begin{proof}
Let $M$ be a term and let $C[\cdot] \ass \bd\gd.(\bd\gd_1.\cdots(\bd\gd_m.[\cdot]\ap\pi)\ap\pi_m\cdots)\ap\pi_1$ be a term context. Let
\[
\begin{array}{ll}
M' \ass M\sub{\pi_1}{\gd_1}\cdots\sub{\pi_m}{\gd_m} \qquad \qquad 
 & \pi' \ass \pi\sub{\pi_1}{\gd_1}\cdots\sub{\pi_m}{\gd_m} \\
U' \ass U\sub{\pi'}{\ga} 
 & \pi'' \ass (\bd\gb.\cadr{\pi'}{0}\ap U'\at M'\at \pi') \at U' \at \pi' \\
\end{array}
\]
Then $C[W[M]] \mslabelto{o} \bd\gd.\cadr{\pi'}{0}\ap \pi''$. At this point we distinguish six possible cases:
\begin{enumerate}[(1)]
\item $\cadr{\pi'}{0}$ is not an abstraction;
\item $\cadr{\pi'}{0}\equiv \bd\epsilon.\cadr{\nil}{n}\ap\varpi$;
\item $\cadr{\pi'}{0}\equiv \bd\epsilon.\cadr{\epsilon'}{n}\ap\varpi$, with $\epsilon \not\equiv \epsilon'$;
\item $\cadr{\pi'}{0}\equiv \bd\epsilon.\cadr{\epsilon}{0}\ap\varpi$;
\item $\cadr{\pi'}{0}\equiv \bd\epsilon.\cadr{\epsilon}{1}\ap\varpi$;
\item $\cadr{\pi'}{0}\equiv \bd\epsilon.\cadr{\epsilon}{n}\ap\varpi$, with $n \geq 2$.
\end{enumerate}
According to the above cases, the outer reduction of $C[W[M]]$ proceeds as:
\begin{enumerate}[(1)]
\item $\ldots \mslabelto{o} \bd\gd.\cadr{\pi'}{0}\ap \pi''\not\labelto{o}$.
\item $\ldots \mslabelto{o} \bd\gd.\cadr{\nil}{n}\ap\varpi\sub{\pi''}{\epsilon} \not\labelto{o} $.
\item $\ldots \mslabelto{o} \bd\gd.\cadr{\epsilon'}{n}\ap\varpi\sub{\pi''}{\epsilon} \not\labelto{o} $.
\item $\ldots \mslabelto{o} \bd\gd.(\bd\epsilon.\cadr{\epsilon}{0}\ap\varpi)\ap \pi'' \mslabelto{o} \bd\gd.\cadr{\pi'}{0}\ap\pi'$.
\item $\ldots \mslabelto{o} \bd\gd.(\bd\epsilon.\cadr{\epsilon}{1}\ap\varpi)\ap \pi'' \mslabelto{o} \bd\gd.\cadr{\pi'}{0}\ap\pi'$.
\item $\ldots \mslabelto{o} \bd\gd.(\bd\epsilon.\cadr{\epsilon}{n}\ap\varpi)\ap \pi'' \mslabelto{o} \bd\gd.\cadr{\pi'}{n-2} \ap\varpi\sub{\pi''}{\epsilon}$.
\end{enumerate}
Suppose $C[W[M]]$ has a proper onf, say, $Z$. In each of the above cases there is no step in the outer reduction path $C[W[M]] \mslabelto{o} Z$
 such that a substitution instance of the occurrence of $M$ put into the hole is active part of a contracted redex. Therefore an isomorphic
 outer-reduction path takes any other term $C[W[N]]$ to its onf, which must be proper too.
\end{proof}


\textbf{Failure of B\"{o}hm's Theorem.} Theorem \ref{thm:wrapper-exists} implies that B\"{o}hm's theorem fails (and quite violently)
 in the stack calculus. Every pair $M,N$ of distinct $\labelto{\s\eta}$-normal forms yields a pair $W[M],W[N]$ of distinct
 $\labelto{\s\eta}$-normal forms which are operationally equivalent.
 
A \emph{stack-theory} (resp. \emph{extensional stack-theory}) for the stack calculus is any set $\cT$ of equalities between stack-expressions
 containing $=_{\s}$ (resp. $=_{\s\eta}$) and closed under context formation and replacement of $\cT$-equal sub-expressions. We indicate by 
 $=_{\cT}$ the congruence associated to the theory $\cT$.

A stack-theory $\cT$ is \emph{inconsistent} if for every pair of terms $M,N$ and every variable $\gc\not\in \FV(M) \cup \FV(N)$ we have that
 $\bd\gc.M\ap\gc =_\cT \bd\gc.N\ap\gc$; $\cT$ is \emph{consistent} otherwise. Since there are distinct $\labelto{\s\eta}$-normal forms,
 Theorem \ref{thm:CR-bd} implies that $=_{\s}$ and $=_{\s\eta}$ are consistent equational theories, of which $=_{\s\eta}$ is extensional.
 Also the relation $\opeq$ is an extensional equational theory which is consistent, because for example $\bT\not\opeq\Omega$.

A theory is \emph{Hilbert\textendash Post complete} (HP-complete, for short) if it is maximally consistent (cannot be properly extended to a consistent theory). The next
 theorem shows that operational equivalence is maximally consistent.

\begin{theorem}\label{thm:HP-complete}
The relation $\opeq$ is an HP-complete equational theory for the stack calculus.
\end{theorem}

\begin{proof}
The relation $\opeq$ is indeed an equational theory for the stack calculus, because it is closed w.r.t. context formation.
 Suppose, by contradiction, that $\cT$ is a consistent theory that contains $\opeq$ properly and let $M,N$ be terms such that
 $M \not\opeq N$ and $M =_{\cT} N$.
 Then there exists a head context $C[\cdot]$ such that, say, $C[M]$ has a proper onf and $C[N]$ does not have a proper onf.
 Assume $C[M] =_{\s} \bd\ga.\cadr{\gb}{n}\ap \veN \at \cddr{\gc}{k}$ and define the context
 $C'[\cdot] \ass \bd\epsilon.(\bd\gb.[\cdot]\ap \ga)\ap \underbrace{\bI\at\ldots\at\bI}_{n}\at (\bd \gd.\bI \ap \epsilon)\at\epsilon$.
 Then $C'[C[M]] =_{\s} \bI$. On the other hand $C'[C[N]]$ cannot have a proper onf.
 Let $u \ass \bd x.\cadr{f}{0}\ap(\bd\gb.\cadr{x}{0}\ap\cadr{x}{0}\at\gb)\at\cddr{x}{1}$, $U \ass \bd\gc.u\ap u\at\gc$ and 
 $Y \ass \bd f.U\ap \cddr{f}{1}$. Finally set $\bT_\infty \ass \bd\gd.Y\ap\bT\at\gd$.

Since $\bT_\infty$ does not have an onf we have $C'[C[N]] \opeq \bT_\infty$ and from $\bI =_{\s} C'[C[M]] =_{\cT} C'[C[N]]$ and
 the fact that the congruence ${=_\cT}$ extends both $=_{\s}$ and $\opeq$, we get that $\bI =_{\cT} \bT_\infty$. Now take an
 arbitrary term $Z$ and a variable $\gc\not\in\FV(Z)$. We have
 $\bd\gc.Z\ap\gc =_{\s} \bd\gc.\bI \ap Z\at\gc =_{\cT} \bd\gc.\bT_\infty \ap Z\at\gc =_{\s} \bT_\infty$.
 Since the congruence ${=_\cT}$ extends both $=_{\s}$ and $\opeq$, we can conclude that $\bd\gc.Z\ap\gc =_{\cT} \bd\gc.Z'\ap\gc$ for all terms
 $Z,Z'$ and every variable $\gc\not\in \FV(Z) \cup \FV(Z')$. Therefore $\cT$ is inconsistent.
\end{proof}

\section{The extended stack calculus}\label{sec:extended-calc}

The \emph{extended stack calculus} is a super-language of the stack calculus. Formally, it is obtained by incorporating the syntactic
 category of processes into that of terms. Therefore the grammar for the extended language is the following one:
\[
\begin{array}{llllr}
\text{stacks} & \quad \pi,\varpi & \gramm & \nil \mid \ga \mid \cdr \pi \mid M \at \pi & \qquad\qquad\qquad\qquad\qquad  \\
\text{terms}  & \quad M,N        & \gramm & \car \pi \mid \bd \ga.M \mid M \ap \pi & 
\end{array}
\]

We still use letters $E,E'$ to range over expressions which are either stacks or terms. An example of term which belongs to the extended
 language but not to the original one is $\bd \gb.\bd \ga.\car\ga \ap \cdr\ga$. Application associates to the left, so that $M \ap \pi \ap \varpi$ stands for
 $(M \ap \pi) \ap \varpi$ and application has precedence over $\bd$-abstraction. If $\seq\pi = \pi_1,\ldots,\pi_m$ and $\seq \ga = \ga_1,\ldots,\ga_n$, we shall
 abbreviate the term $\bd\ga_1\ldots\bd\ga_n.(M \ap \pi_1 \ap \cdots \ap \pi_m)$ as $\bd\seq\ga.M\ap \seq\pi$. We denote by $\KTer{t}$, $\KTer{s}$
 and $\KTer{e}$ the sets of all terms, stacks and expressions respectively. We still use $\labelto{\s}$, $\labelto{\s\eta}$, $=_{\s}$ and $=_{\s\eta}$ to indicate the straightforward extensions of the corresponding relations
 defined for the stack calculus. Also the notion of canonical form extends straightforwardly to the extended stack calculus.
 
\begin{theorem}\label{thm:CR-bd}
The reductions $\labelto{\s}$ and $\labelto{\s\eta}$ in the extended stack calculus are both Church-Rosser.
\end{theorem}

The proof of Theorem \ref{thm:CR-bd} is an easy modification of that of Theorem \ref{thm:CR-bd-original} (see \cite{CaEhSa12}).

Also the definitions of (consistent) equational theories and HP-completeness extend straightforwardly to the larger calculus and once again Theorem \ref{thm:CR-bd}
 guarantees the consistency of the theories $=_{\s}$ and $=_{\s\eta}$.
 
In this paper we are not concerned with semantics. We just mention that the works of Streicher and Reus \cite{Streicher98} and
 Nakazawa and Katsumata \cite{Naka12} already provide sound models for the extended stack calculus, which are the same as those
 for the $\Lambda\mu$-calculus.




\subsection{Head-reduction and operational equivalence for the extended calculus}\label{subsec:head-red}

We now provide a notion of \emph{head-reduction} for the stack calculus, which is performed by contracting the left-most redex only.
 The one-step \emph{head-reduction} on terms is given by the following rule:

$$ M \labelto{h} \bd\seq \ga.N\sub{\varpi}{\gb} \ap \seq \pi \qquad \text{ if } \bd\seq \ga.(\bd \gb.N) \ap \varpi \ap \seq \pi \text{ is the canonical form of } M $$

Note that head-reduction is \emph{deterministic} and we don't take any contextual closure for $\labelto{h}$; its reflexive and
 transitive closure is denoted by $\mslabelto{h}$. A term $M$ \emph{is in head-normal form} (\emph{hnf}, for short) if it is not
 $\labelto{h}$-reducible. It is straightforward to see that a term is in hnf iff it has the form $\bd\seq\ga.H\ap \seq\pi$, where the stacks in the
 sequence $\seq \pi$ are arbitrary, $H$ is either $\cadr{\gb}{n}$ (in which case we say that the hnf is \emph{proper}) or $\cadr{\nil}{n}$
 (in which case we say that the hnf is \emph{improper}); $H$ is the head of the hnf in question.
 A term \emph{has} a hnf if it $\mslabelto{h}$-reduces to a term in hnf. Of course the head-reduction strategy is complete for finding
 hnf's of terms, i.e., if $M$ $\mslabelto{\s}$-reduces to a term $N'$ in hnf, then $M$ reduces to some hnf $N''$ via head-reduction.
 
For convenience we define a partial function \mbox{$\hnf{\cdot}: \KTer{t} \rightharpoonup \KTer{t}$} which returns the $\labelto{h}$-normal form of
 a term, if it exists.

The \emph{head contexts} of the extended stack calculus are produced by the following grammar:
$$ C[\cdot] \gramm [\cdot] \mid C[\cdot] \ap \pi \mid \bd\ga.C[\cdot] $$

The next definition is the analogue of Definition \ref{def:op-equiv} for the extended stack calculus.

\begin{definition}[Operational equivalence (extended)]\label{def:op-equiv-extended}
Two terms $M,N$ are \emph{operationally equivalent}, notation $M \opeq N$, if for every head context $C[\cdot]$ we have that
 $C[M]$ has a proper hnf iff $C[N]$ has a proper hnf.
\end{definition}
 
We use the same symbol as in Definition \ref{def:op-equiv} because there will be no ambiguity: from now on we are only concerned with the extended
 stack calculus. The relation $\opeq$ is a consistent extensional theory which is HP-complete, because the proof of Theorem \ref{thm:HP-complete}
 works also for the larger calculus.

\begin{theorem}\label{thm:HP-complete-extended}
The relation $\opeq$ is an HP-complete equational theory for the extended stack calculus.
\end{theorem}

\subsection{Similarity and separability}\label{sec:similarity}

%

The following definition introduces an important concept, somewhat orthogonal to operational equivalence, which is very typical of
 functional calculi (as the $\lambda$-calculus).

\begin{definition}[Separability]\label{def:separability}
We say that $M$ and $N$ are \emph{separable} if there exists a head context $C[\cdot]$ such that \mbox{$C[M] =_{\s} \bT$} and
 \mbox{$C[N] =_{\s} \bF$}.
\end{definition}

The following theorem says that separability and operational equivalence are somewhat orthogonal to each other.

\begin{theorem}\label{thm:sepa-nopeq}
If $M$ and $N$ are separable, then $M \not\opeq N$.
\end{theorem}

\begin{proof}
Suppose $M,N$ separable. Then there exists a context $C[\cdot]$ such that $C[M] =_{\s} \bT$ and $C[N] =_{\s} \bF$ and setting
 $C'[\cdot] \ass \bd\epsilon.C[\cdot]\ap\Omega\at\bI\at\epsilon$ we obtain that $C'[M] =_{\s} \Omega$ and $C'[N] =_{\s} \bI$, thus showing
 that $M \not\opeq N$.
\end{proof}

The converse of Theorem \ref{thm:sepa-nopeq} does not hold. For example $\Omega \not\opeq \bI$ but it is also true that
 $\Omega$ and $\bI$ are not separable. In fact if a term $M$ does not have a hnf, then $C[M]$ does not have a hnf too for every head context
 $C[\cdot]$. This means that for no head context $C[\cdot]$ we can have $C[\Omega] =_{\s} \bT$ (or $C[\Omega] =_{\s} \bF$).

Next we introduce the notion of \emph{similarity} between stacks and terms which, for terms having an hnf, is weaker than inseparability. 

\begin{definition}[Similarity for stacks]\label{def:similarity-stacks}
We define the similarity relation $\ssim$ on $\KTer{s}$ as the smallest equivalence relation closed under $=_{\s}$ satisfying the
 following conditions:
\begin{enumerate}[(1)]
\item $\pi \ssim M_1\at \ldots \at M_m \at \cddr{\nil}{k}$
\item if $k-m = k'-m'$, then $M_1\at \ldots \at M_m \at \cddr{\gc}{k} \ssim N_1\at \ldots \at N_{m'} \at \cddr{\gc}{k'}$
\end{enumerate}
\end{definition}

\begin{definition}[Similarity for terms]\label{def:similarity}
We define the similarity relation $\sim$ on $\KTer{t}$ as the smallest equivalence relation closed under $=_{\s}$ satisfying the
 following conditions:
\begin{enumerate}[(1)]
\item if $k-m = k'-m'$, $\pi_i \ssim \varpi_i$ for all $i = 1,\ldots,min\{m,m'\}$, and
 \mbox{$\varpi_{min\{m,m'\}+j}\ssim\ga_{min\{k,k'\}+j}$} for all $j = 1,\ldots,(max\{k,k'\}-min\{k,k'\})$ then
$$ \bd\ga_1\ldots\ga_k.\cadr{\gb}{n} \ap \pi_1 \aps \pi_m \sim \bd\ga_1\ldots\ga_{k'}.\cadr{\gb}{n} \ap \varpi_1 \aps \varpi_{m'} $$
\item if $\hnf{M}$ and $\hnf{N}$ are both defined and improper, then $M \sim N$
\item if $\hnf{M}$ and $\hnf{N}$ are both undefined, then $M \sim N$
\end{enumerate}
\end{definition}


The following theorems show that dissimilarity of terms having a hnf implies their separability. We start by treating a particular
 case, which is nevertheless non-trivial.

It will occur frequently to use stacks of the form $\overbrace{M\at\ldots\at M}^{n}\at\pi$. Therefore we set the special notation
 $M^{n}\at\pi$ for these stacks.

\begin{theorem}\label{thm:10.4.1-particular}
Let $M \equiv \bd\ga.\cadr{\gb}{n}\ap\pi$ and $N \equiv \bd\ga.\cadr{\gb'}{n'}\ap\pi'$. If $M \not\sim N$, then $M$ and $N$ are separable.
\end{theorem}

\begin{proof}
Assume $M \not\sim N$. We analyze the different reasons for this fact and each time we build a context $C[\cdot]$ such that $C[M] =_{\s} \bT$ and
 $C[N] =_{\s} \bF$. In the rest of the proof we let $\epsilon$ be a fresh variable.
 Since $M$ and $N$ are not similar, we have the following possible cases:
\begin{itemize}
\item[\textbf{(1)}] $\gb \not\equiv \gb'$;
\item[\textbf{(2)}] $\gb \equiv \gb'$ but $n \neq n'$;
\item[\textbf{(3)}] $\gb \equiv \gb'$, $n = n'$ but $\pi \not\ssim \pi'$.
\end{itemize}

\noindent\textbf{(1)} Define $\pi \ass \bI^{n}\at (\bd\gd.\bT \ap\epsilon) \at \epsilon$, where $\gd \not\equiv \epsilon$,\ 
 \mbox{$\pi'\ass\bI^{n'}\at(\bd\gd.\bF \ap\epsilon)\at\epsilon$}, where $\gd \not\equiv \epsilon$, and 
 $C[\cdot] \ass \bd \epsilon.(\bd \gb'.(\bd \gb.[\cdot] \ap \ga) \ap \pi) \ap \pi'$. Then $C[M] =_{\s} \bT$ and $C[N] =_{\s} \bF$. \\

\noindent\textbf{(2)} We can assume w.l.o.g. that $n > n'$ since the opposite case can be treated symmetrically. Define
 \mbox{$\pi \ass \bI^{n'}\at(\bd\gd.\bF\ap\epsilon)\at\bI^{n-n'-1}\at(\bd\gd.\bT\ap\epsilon)\at\epsilon$},
 where $\gd \not\equiv \epsilon$, and $C[\cdot] \ass \bd \epsilon.(\bd \gb.[\cdot] \ap \ga) \ap \pi$. Then $C[M] =_{\s} \bT$ and $C[N] =_{\s} \bF$. \\

In order to treat case \textbf{(3)}, we need to explicit the possible forms of $\pi$ and $\pi'$. By assumption there some are terms
 $M_1,\ldots, M_m,N_1,\ldots,N_{m'}$ and numbers $k,k'\in \nat$ such that:
\begin{itemize}
\item $M \equiv \bd\ga.\cadr{\gb}{n} \ap M_1\at \ldots \at M_m \at \cddr{\gc}{k}$
\item $N \equiv \bd\ga.\cadr{\gb'}{n'} \ap N_1\at \ldots \at N_{m'} \at\cddr{\gc'}{k'}$
\end{itemize}
Case \textbf{(3)} can be exhaustively splitted in the following sub-cases:
\begin{itemize}
\item[\textbf{(3.1)}] $\gb \equiv \gb'$, $n = n'$ but $\gc \not\equiv \gc'$ (no assumptions on $m,k,m',k'$) and
\begin{itemize}
\item[\textbf{(3.1.1)}] $\gc \not\equiv \gb$;
\item[\textbf{(3.1.2)}] $\gc \equiv \gb$.
\end{itemize}
\item[\textbf{(3.2)}] $\gb \equiv \gb'$, $n = n'$, $\gc \equiv \gc'$ but $m-k \neq m'-k'$ and
\begin{itemize}
\item[\textbf{(3.2.1)}] $\gc \not\equiv \gb$;
\item[\textbf{(3.2.2)}] $\gc \equiv \gb$.
\end{itemize}
\end{itemize}

We now show how to reduce the cases (3.1.1), (3.1.2) and (3.2.1) to the case (3.2.2), for which we show how to build the separating context. \\

\noindent\textbf{(3.1.1)} If $\gc' \equiv \gb$, then case \textbf{(3.1.2)} applies (changing the roles of $\gc$ and $\gc'$).
 If $\gc' \not\equiv \gb$, then define $C[\cdot] \ass \bd\ga.(\bd\gc.[\cdot] \ap \ga) \ap \gb$. Now case \textbf{(3.1.2)} applies to $C[M]$ and
 $C[N]$. \\

\noindent\textbf{(3.1.2)} Define \mbox{$C[\cdot] \ass \bd\ga.(\bd\gc'.[\cdot] \ap \ga) \ap \cddr{\gc'}{m'+m+k+1}$}. Now case \textbf{(3.2.2)}
 applies to $C[M]$ and $C[N]$. \\
 
\noindent\textbf{(3.2.1)} Define $C[\cdot] \ass \bd\ga.(\bd\gc.[\cdot] \ap \ga) \ap \gb$. Now case \textbf{(3.2.2)} applies to $C[M]$ and $C[N]$. \\

\noindent\textbf{(3.2.2)} Let $\partial m = |m-m'|$, $\partial k = |k-k'|$, $d = \partial k + \partial m$, and
 $e = min\{\partial k,\partial m\}$. Note that $m-k = m'-k'$ if, and only if, $m-m' = k-k'$. Therefore, under the hypothesis of this case,
 $m-m' \neq k-k'$ and $d >0$. 
$$
X \ass
\begin{cases}
\bd\gd.\cadr{\gd}{0}\ap \cddr{\epsilon}{2} & \text{if $k < k'$} \\
\bd\gd.\cadr{\gd}{e}\ap \cddr{\epsilon}{2} & \text{if $k \geq k'$} \\
\end{cases}
$$
Moreover define 
\mbox{$\pi \ass \bI^{n}\at(\bd\gd.X\ap \cddr{\gd}{n+1+max \{m,m'\}})\at{\cadr{\epsilon}{1}}^{\partial m+max \{k,k'\}} \at \cadr{\epsilon}{0}\at\epsilon$}
 and\\ \mbox{$C[\cdot] \ass \bd \epsilon.(\bd \gb.[\cdot] \ap \ga) \ap \pi$}. We can assume w.l.o.g. that $m \geq m'$ (the other case can be treated symmetrically)
 so that we have
$$
C[M] =_{\s}
\begin{cases}
\bd \epsilon.X\ap \cadr{\epsilon}{1}^{d} \at \cadr{\epsilon}{0}\at\epsilon & \text{if $k < k'$} \\
\bd \epsilon.X\ap \cadr{\epsilon}{1}^{\partial m}\at\cadr{\epsilon}{0}\at \epsilon & \text{if $k \geq k'$} \\
\end{cases}
\quad \text{ and } \quad
C[N] =_{\s}
\begin{cases}
\bd \epsilon.X\ap \cadr{\epsilon}{0}\at \epsilon & \text{if $k < k'$} \\
\bd \epsilon.X\ap \cadr{\epsilon}{1}^{\partial k} \at\cadr{\epsilon}{0}\at \epsilon & \text{if $k \geq k'$} \\
\end{cases}
$$
Concluding the computation we obtain
$$
C[M] =_{\s}
\begin{cases}
\bF & \text{if $k < k'$} \\
\bT & \text{if $k \geq k'$} \\
\end{cases}
\qquad \text{ and } \qquad
C[N] =_{\s}
\begin{cases}
\bT & \text{if $k < k'$} \\
\bF & \text{if $k \geq k'$} \\
\end{cases}
$$
\end{proof}

Now we can now prove the following more general statement.

\begin{theorem}\label{thm:10.4.1}
Let $M,N$ be terms having a proper hnf. Then $M \not\sim N$ implies that $M$ and $N$ are separable.
\end{theorem}

\begin{proof}
Suppose $M$ and $N$ be terms with a proper hnf and assume $M \not\sim N$. We analyze the different cases in which $M \not\sim N$ and each time
 we build a context $C[\cdot]$ such that $C[M] =_{\s} \bT$ and $C[N] =_{\s} \bF$. Since $M,N$ both have a hnf, let's say we have
\begin{itemize}
\item $\hnf{M} = \bd\ga_1\ldots\ga_k.\cadr{\gb}{n} \ap \pi_1 \aps \pi_m$
\item $\hnf{N} = \bd\ga_1\ldots\ga_{k'}.\cadr{\gb'}{n'} \ap \varpi_1 \aps \varpi_{m'}$
\end{itemize}
In the rest of the proof we let $\epsilon$ be a fresh variable and we suppose, w.l.o.g. that $k'\geq k$. Since $M$ and $N$ are not
 similar, we have the following possible cases:
\begin{itemize}
\item[\textbf{(1)}] $\gb \not\equiv \gb'$;
\item[\textbf{(2)}] $\gb \equiv \gb'$ but $n \neq n'$;
\item[\textbf{(3)}] $\gb \equiv \gb'$, $n = n'$ but $m-k \neq m'-k'$;
\item[\textbf{(4)}] $\gb \equiv \gb'$, $n = n'$, $m-k = m'-k'$ but there is some $i = 1,\ldots,m$ such that $\pi_i \not\ssim \varpi_i$ (note that $m'\geq m$);
\item[\textbf{(5)}] $\gb \equiv \gb'$, $n = n'$, $m-k = m'-k'$, $\pi_i \ssim \varpi_i$ for all $i = 1,\ldots,m$ but there is some
 $j = 1,\ldots,k'-k$ such that $\varpi_{m+j} \not\ssim \ga_{k+j}$ (note that $m'\geq m$).
\end{itemize}

We now show how to build in each case a separating context. \\

\noindent\textbf{(1)} Define \mbox{$\pi \ass \bI^{n}\at (\bd\gd_1\ldots\gd_m\ga_{k+1}\ldots\ga_{k'}.\bT)\at\epsilon$},\ 
 \mbox{$\pi' \ass \bI^{n'}\at (\bd\gd_1\ldots\gd_{m'}.\bF)\at\epsilon$}, and \\
 \mbox{$C[\cdot] \ass (\bd\gb\gb'.[\cdot] \ap \ga_1 \aps \ga_{k'}) \ap \pi \ap \pi'$}. Then $C[M] =_{\s} \bT$ and $C[N] =_{\s} \bF$. \\

\noindent\textbf{(2)} Suppose that $n'> n$ (all other cases can be treated similarly). Define\\
 \mbox{$\pi \ass \bI^{n}\at(\bd\gd_1\ldots\gd_m\ga_{k+1}\ldots\ga_{k'}.\bT)\at\bI^{n'-n-1}\at(\bd\gd_1\ldots\gd_{m'}.\bF)\at\epsilon$}
 and \mbox{$C[\cdot] \ass (\bd \gb.[\cdot] \ap \ga_1 \aps \ga_{k'}) \ap \pi$}. Then $C[M] =_{\s} \bT$ and $C[N] =_{\s} \bF$. \\

\noindent\textbf{(3)} Note that, having assumed $m-k \neq m'-k'$, we have $m + k'-k \neq m'$. Suppose $m + k'-k > m'$ (the opposite
 case can be treated similarly). Now let $p = m + k'-k$ , $h = p - m'$ and define \\
 \mbox{$\pi \ass \bI^{n}\at (\bd\ga_1\ldots\ga_{p+1}.\cadr{\ga_{p+1}}{0}) \at\epsilon$},
 \mbox{$C[\cdot] \ass (\bd \gb.[\cdot] \ap \ga_1 \aps \ga_{k'} \ap \gd \ap \epsilon_1 \aps \epsilon_h) \ap \pi$}, and \\
 \mbox{$C'[\cdot] \ass (\bd\gd\epsilon_h.[\cdot])\ap ((\bd\ga_1\ldots\ga_{h}.\bT)\at\epsilon)\ap(\bF\at\epsilon)$}.
 Then $C'[C[M]] =_{\s} \bT$ and $C'[C[N]] =_{\s} \bF$. \\

\noindent\textbf{(4)} Define \mbox{$X \ass \bd\ga_1\ldots\ga_{m'}\gb''.\cadr{\gb''}{0}\ap \ga_i$},
 \mbox{$\pi' \ass \bI^{n}\at X \at\epsilon$}, and
 \mbox{$C[\cdot] \ass (\bd \gb.[\cdot] \ap \ga_1 \aps \ga_{k'}) \ap \pi'$}. Then we have
$$
 C[M] =_{\s} \bd\gb''.\cadr{\gb''}{0}\ap \pi_i\sub{\pi'}{\gb} 
 \text{ and } 
C[N] =_{\s} \bd\gb''.\cadr{\gb''}{0}\ap \varpi_i\sub{\pi'}{\gb}
$$
Finally since $\pi_i \not\ssim \varpi_i$ and $\not\ssim$ is closed under substitution, we also have
 \mbox{$\pi_i\sub{\pi'}{\gb}\not\ssim \varpi_i\sub{\pi'}{\gb}$}, which in turn implies $C[M] \not\sim C[N]$.
 At this point we can apply Theorem \ref{thm:10.4.1-particular} to show that $C[M]$ and $C[N]$ are separable, and hence $M$ and $N$ are separable.\\

\noindent\textbf{(5)} Define $X \ass \bd\ga_1\ldots\ga_{m'}\gb''.\cadr{\gb''}{0}\ap \ga_{m+j}$,\ 
 $\pi' \ass \bI^{n}\at X \at\epsilon$, and 
 \mbox{$C[\cdot] \ass (\bd \gb.[\cdot] \ap \ga_1 \aps \ga_{k'}) \ap \pi'$}. Then we have
$$
C[M] =_{\s} \bd\gb''.\cadr{\gb''}{0}\ap \ga_{k+j}\sub{\pi'}{\gb}
 \text{ and }
C[N] =_{\s} \bd\gb''.\cadr{\gb''}{0}\ap \varpi_{m+j}\sub{\pi'}{\gb}
$$
Finally since $\varpi_{m+j} \not\ssim \ga_{k+j}$ and $\not\ssim$ is closed under substitution, we also have
 $\varpi_{m+j}\sub{\pi'}{\gb} \not\ssim \ga_{k+j}\sub{\pi'}{\gb}$, which in turn implies $C[M] \not\sim C[N]$.
 At this point we can apply Theorem \ref{thm:10.4.1-particular} to show that $C[M]$ and $C[N]$ are separable, and hence $M$ and $N$ are separable.
\end{proof}


The converse of Theorem \ref{thm:10.4.1} does not hold; for example $\bd\gc.\cadr{\gc}{0}\ap \bT\at\gc$ and $\bd\gc.\cadr{\gc}{0}\ap \bF\at\gc$ are
 separable but it is also true that $\bd\gc.\cadr{\gc}{0}\ap \bT\at\gc \sim \bd\gc.\cadr{\gc}{0}\ap \bF\at\gc$. The point is that the relation $\sim$
 only looks at the ``surface'' of terms, while separation may require to unravel terms by iteratively computing head normal forms, i.e., looking at
 their B\"{o}hm trees.

\subsection{B\"{o}hm out technique and B\"{o}hm's theorem}\label{subsec:Bohm-out}

As for the $\lambda$-calculus, B\"{o}hm trees can be defined for stack terms too. While the nodes of B\"{o}hm trees of $\lambda$-terms
 are indexed by sequences of natural numbers, the nodes of B\"{o}hm trees of stack terms should be indexed by sequences of pairs of natural numbers.
 The reason for this choice is that a child of the node corresponding to a hnf $\bd\seq\ga.\cadr{\gb}{n}\ap\seq \pi$ must be
 selected with two coordinates $(j,j')$, saying that the child is the root of the B\"{o}hm tree of the $j'$-th term of the canonical form
 of the stack $\pi_{j}$.

The B\"{o}hm trees for the $\lambda$-calculus (see \cite[\S~10]{Bare84}) are partial functions mapping sequences of natural numbers
 either to a special constant $\bot$ or to a $\lambda$-term o shape $\lambda\seq x.y$. We will still speak of B\"{o}hm trees for the extended 
 stack calculus, but strictly speaking we mean partial functions mapping pairs of natural numbers to non-necessarily normal terms, called \emph{nodes}.
 Letters $\gs,\gt,\gr,\ldots$ range over the set $\Seq$ of all finite sequences of pairs of strictly positive natural numbers. We define the order
 $<$ on these sequences as follows: $\mbox{$\gt < \gs$ iff $\gt$ is a proper prefix of $\gs$}$. We denote by $\leng \sigma$ the length of
 the sequence $\sigma$.
 Moreover if $\seq N$ is a sequence of terms, then $\leng \seq N$ indicates the lenght of $\seq N$.

\begin{definition}\label{def:10.1.13}
Given a term $M$ we define a partial map $M(\cdot): \Seq \rightharpoonup \KTer{t}$ as follows:
\[
\begin{array}{l}
M(\gs) \ass M \qquad \text{ if $\sigma$ is the empty sequence} \\
M(\gt \cdot (j,j')) \ass 
\begin{cases}
N_{j'}           & \text{if $M(\gt)$ is defined, $\hnf{M(\gt)} = \bd\seq\ga.\cadr{\gb}{n}\ap \pi_1\aps\pi_m$ and} \\
                 & \text{$j \leq m$ and $\pi_j$ has canonical form $\seq N\at\cddr{\gc}{k}$ or $\seq N\at\cddr{\nil}{k}$, with $j' \leq \leng \seq N$} \\
                 &  \\
\text{undefined} & \text{otherwise}
\end{cases}
\end{array}
\]
\end{definition}

The map $M(\cdot) : \Seq \rightharpoonup \KTer{t}$ for us \emph{is} the B\"{o}hm tree of $M$. We let $\dom{M} = \{\gs \in \Seq \sut M(\gs) \text{ is defined } \}$.

In the present section we prove a theorem which is the analogue of B\"{o}hm's Theorem for $\lambda$-calculus. Such theorem
 is supported fundamentally by the forthcoming Lemma \ref{lem:10.3.7bis} (analogous to what is called the B\"{o}hm out Lemma for the $\lambda$-calculus \cite[\S~10]{Bare84})
 which shows how to extract substitution instances of nodes of the B\"{o}hm tree of a term, in such a way that some important properties are preserved.

For technical reasons we need to introduce the set $\bdom{M}{n} = \{\gs \in \dom{M} \sut \leng \gs \leq n \}$. The following definitions \ref{def:10.1.13b}, \ref{def:10.1.13c} and \ref{def:special-term-stack}
 are all auxiliary for the statement and proof of the subsequent Lemma \ref{lem:10.3.7}.

\begin{definition}\label{def:10.1.13b}
We define the \emph{breadth} $\brea{M}$ and the \emph{weight} $\wei{M}$ of $M \in \KTer{t}$ as
$$
\brea{M} =
\begin{cases}
0 & \text{ if $\hnf{M}$ is undefined} \\
0 & \text{ if $\hnf{M}$ is defined but improper} \\
m & \text{ if $\hnf{M} = \bd\seq\ga.\cadr{\gb}{n}\ap \pi_1\aps\pi_m$} \\
\end{cases}
\quad \ 
\wei{M} =
\begin{cases}
0 & \text{ if $\hnf{M}$ is undefined} \\
0 & \text{ if $\hnf{M}$ is defined but improper} \\
n & \text{ if } \hnf{M} = \bd\seq\ga.\cadr{\gb}{n}\ap \pi_1\aps\pi_m \\
\end{cases}
$$
\end{definition}

\begin{definition}\label{def:10.1.13c}
The $n$-\emph{bounded breadth} $\bbrea{M}{n}$ and the $n$-\emph{bounded weight} $\bwei{M}{n}$ of a term $M$ are defined as
 \mbox{$\bbrea{M}{n} = max \{ \brea{M(\gs)} \sut \gs \in \bdom{M}{n} \}$} and 
 \mbox{$\bwei{M}{n} = max \{ \wei{M(\gs)} \sut \gs \in \bdom{M}{n} \}$}, respectively.
\end{definition}

%

\begin{definition}\label{def:special-term-stack}
Let $q,p$ be natural numbers. We define the expressions
$$
 \spt{q} \ass \bd\epsilon_1\ldots\epsilon_q\gd.\cadr{\gd}{0}\ap\epsilon_1\aps\epsilon_q
 \qquad\qquad 
 \sps{\epsilon}{q}{p} \ass \underbrace{\spt{q}\at\ldots\at\spt{q}}_{p \text{ times}} \at\epsilon
$$
\end{definition}

Clearly by the common conventions about bound variables, in the above definition $\epsilon_1,\ldots,\epsilon_q,\gd$ are all distinct from each other
 and from $\epsilon$.


The next lemma is the combinatorial core of the forthcoming Lemma \ref{lem:10.3.7bis}, and is the analogue for the extended stack
 calculus of what is called the \emph{B\"{o}hm-out technique} for the $\lambda$-calculus (see \cite[\S~10]{Bare84}).
 For a sequence $\gs = (j_1,j_1')\cdots (j_l,j_l')$ and a natural number $i \leq l$ we indicate with $\trun{\gs}{i}$ the sequence
 $(j_1,j_1')\cdots (j_i,j_i')$ (so for example $\trun{\gs}{0}$ is the empty sequence).

\begin{lemma}[B\"{o}hm out]\label{lem:10.3.7}
Let $M$ be a term, let $n$ a natural number and let $q \geq \bbrea{M}{n}$, $p\geq \bwei{M}{n}$.
 Then for every sequence $\gs\in \bdom{M}{n}$ there exists a context $C[\cdot]$ such that \\
 $C[M] \mslabelto{\s} M(\gs)\sub{\sps{\epsilon_1}{q}{p}}{\gb_1}\cdots\sub{\sps{\epsilon_{l}}{q}{p}}{\gb_{l}}$ where
\begin{enumerate}[(i)]
\item $l = \leng \gs$ and for each $i=1,\ldots,l$, $\gb_i$ is the head variable of\\
      $M({\trun{\gs}{i-1}})\sub{\sps{\epsilon_1}{q}{p}}{\gb_1}\cdots\sub{\sps{\epsilon_{i-1}}{q}{p}}{\gb_{i-1}}$
\item $\epsilon_1,\ldots,\epsilon_{l}$ is a sequence of pairwise distinct variables disjoint from $\gb_1,\ldots,\gb_{l}$ and not occurring in $M$
\item $C[M]$ has a proper hnf iff $M(\gs)$ has a proper hnf.
\end{enumerate}
\end{lemma}

\begin{proof}
We proceed at the same time to define the context $C[\cdot]$ and prove its properties by induction on the length of $\gs$.

\noindent If $\gs$ is the empty sequence, then $C[\cdot] \ass [\cdot]$ and the statement is trivially satisfied. 

\noindent Let $\gs = (j_1,j_1')\cdots (j_{l},j_{l}') \in \bdom{M}{n}$ and $\gt = \trun{\gs}{l-1}$.
 Now $\gt \in \bdom{M}{n}$, since $\gt < \gs$, so by induction hypothesis applied to $\gt$ we have a context $C[\cdot]$ such that\\
 $C[M] \mslabelto{\s} M(\gt)\sub{\sps{\epsilon_1}{q}{p}}{\gb_1}\cdots\sub{\sps{\epsilon_{l-1}}{q}{p}}{\gb_{l-1}}$ where
 \begin{enumerate}[(i)]
 \item for each $i=1,\ldots,l-1$, $\gb_i$ is the head variable of
 \mbox{$M({\trun{\gt}{i-1}})\sub{\sps{\epsilon_1}{q}{p}}{\gb_1}\cdots\sub{\sps{\epsilon_{i-1}}{q}{p}}{\gb_{i-1}}$}
 \item $\seq \epsilon = \epsilon_1,\ldots,\epsilon_{l-1}$ is a sequence of pairwise distinct variables disjoint from
       $\seq \gb = \gb_1,\ldots,\gb_{l-1}$ and not occurring in $M$.
 \end{enumerate}
For simplicity, we write $E^*$ for $E\sub{\sps{\epsilon_1}{q}{p}}{\gb_1}\cdots\sub{\sps{\epsilon_{l-1}}{q}{p}}{\gb_{l-1}}$, so that
 $C[M] \mslabelto{\s} (M(\gt))^*$. The definition of the new context for the longer sequence $\gs$ is based on the shape of the hnf of $C[M]$
 (and therefore on its existence). Since $\gt \in \bdom{M}{n}$ and it is not maximal, we have that $M(\gt)$ does have a proper hnf, say,
 $\hnf{M(\gt)} \equiv \bd\seq\ga.\cadr{\gb}{h}\ap \pi_1\at\ldots\at \pi_m$.
 Then we have
\[
\begin{array}{lcl}
\hnf{M(\gt)}^* 
& \equiv &
\begin{cases}
\bd\seq\ga.\spt{q} \ap \pi_1^*\aps \pi_m^*       & \text{if } \gb \in \seq\gb \\
\bd\seq\ga.\cadr{\gb}{h} \ap \pi_1^*\aps \pi_m^* & \text{if } \gb \not\in \seq\gb
\end{cases}
\\ 
& \mslabelto{\s} &
\begin{cases}
\bd\seq\ga\varepsilon_{m+1}\ldots\varepsilon_{q}\gd.\cadr{\gd}{0}\ap\pi_1^*\aps \pi_m^*\ap \varepsilon_{m+1}\aps\varepsilon_{q} 
    & \text{if } \gb \in \seq\gb \\
\bd\seq\ga.\cadr{\gb}{h} \ap \pi_1^*\aps \pi_m^* 
    & \text{if } \gb \not\in \seq\gb
\end{cases}
\\
& \equiv & \hnf{(M(\gt))^*}
\end{array}
\]
because by hypothesis $p \geq h$ and $q \geq m$. The computations above give, according to the different cases, the head variable of $C[M]$,
 since $C[M] \mslabelto{\s} (M(\gt))^*$. Now let $\epsilon \not\in \FV(\hnf{(M(\gt))^*})$ and set
\begin{itemize}
\item
$C'[\cdot] \ass
\begin{cases}
[\cdot] \ap \seq \ga & \text{if $\gb \in \seq\gb$} \\ 
(\bd\gb.[\cdot] \ap \seq \ga) \ap \sps{\epsilon}{q}{p} & \text{if $\gb \not\in \seq\gb$ and $\gb\not\in \seq\ga$} \\
[\cdot] \ap \ga_1 \aps \ga_{r-1}\ap \sps{\epsilon}{q}{p} \ap \ga_{r+1} \aps \ga_t & \text{if $\gb \not\in \seq\gb$ and $\gb \equiv \ga_r \in \seq\ga = \ga_1,\ldots,\ga_t$}
\end{cases}$
\item $C''[\cdot] \ass [\cdot]\ap \varepsilon_{m+1}\aps\varepsilon_{q} \ap ((\bd\ga_1\ldots\ga_q.\cadr{\ga_{j_l}}{j_l'-1})\at\epsilon)$
\end{itemize}
We claim that the context $D[\cdot] \ass C''[C'[C[\cdot]]]$ satisfies the statement of the lemma. By all the definitions and results above we have that
$$
C'[\hnf{(M(\gt))^*}] \mslabelto{\s} \bd\varepsilon_{m+1}\ldots\varepsilon_{q}\gd.\cadr{\gd}{0}\ap\pi_1^*\sub{\sps{\epsilon}{q}{p}}{\gb}\aps\pi_m^*\sub{\sps{\epsilon}{q}{p}}{\gb}\ap\varepsilon_{m+1}\aps\varepsilon_{q}
$$
because if $\gb \in \seq\gb$, then the sequences of stacks $\pi_1^*,\ldots,\pi_m^*$ and 
 $\pi_1^*\sub{\sps{\epsilon}{q}{p}}{\gb},\ldots,\pi_m^*\sub{\sps{\epsilon}{q}{p}}{\gb}$ coincide since in that case the variable
 $\gb$ does not occur free in $\pi_1^*,\ldots,\pi_m^*$. Therefore 
$$
D[M] \mslabelto{\s} C''[C'[(M(\gt))^*]] \mslabelto{\s} C''[C'[\hnf{(M(\gt))^*}]] \mslabelto{\s} \cadr{(\pi_{j_l}^*\sub{\sps{\epsilon}{q}{p}}{\gb})}{j_l'-1} \mslabelto{\s} (M(\gs))^*\sub{\sps{\epsilon}{q}{p}}{\gb}
$$
Finally we remark that $M(\gs)$ has a proper hnf iff $(M(\gs))^*\sub{\sps{\epsilon}{q}{p}}{\gb}$ has a proper hnf. This concludes the proof.
\end{proof}

Note that if $\sigma \in \dom{M}$ is non-empty, then for every proper prefix $\gt<\gs$, the term $M(\gt)$ must have a proper hnf. This fact 
 allows the ``navigation'' of the B\"{o}hm tree of $M$ implemented in Lemma \ref{lem:10.3.7}. The improper hnf's do not play the same role that
 head normal forms have in the $\lambda$-calculus.

 
Now we want to look at terms as maps which are defined also at nodes reachable by the suitable amount of $\eta$-expansions. To this end the
 following concept of \emph{path expansion} will be used to define these maps.

\begin{definition}[Path expansion]\label{def:path-exp}
Let $\gs$ be a sequence and let $M$ be a term. We define the \emph{path expansion $\pexp{M}{\gs}$ of $M$ by $\gs$} by induction
 on the length of $\gs$ as follows:
\[
\begin{array}{l}
\pexp{M}{\sigma} \ass M \qquad \text{ if $\sigma$ is the empty sequence} \\
\pexp{M}{(j,j')\cdot\gt} \ass 
\begin{cases}
\bd\seq\ga.\cadr{\gb}{n}\ap\pi_1\aps (\seq N\at\cadr{\gc}{k}\ats\pexp{\cadr{\gc}{k+j'-1}}{\gt}\at\cddr{\gc}{k+j'}) \aps \pi_m & \\
\quad\text{if $\hnf{M} = \bd\seq\ga.\cadr{\gb}{n}\ap \pi_1\aps \pi_m$, $j \leq m$ and} & \\
\quad\text{$\seq N\at\cddr{\gc}{k}$ is the canonical form of $\pi_j$ and $j' > \leng \seq N$} & \\
                         & \\
\bd\seq\ga\gc_1\ldots\gc_{j-m}.\cadr{\gb}{n}\ap\pi_1\aps\pi_m\ap\gc_1\aps(\cadr{\gc_{j-m}}{0}\ats\pexp{\cadr{\gc_{j-m}}{j'-1}}{\gt}\at\cddr{\gc_{j-m}}{j'}) & \\
\quad\text{if $\hnf{M} = \bd\seq\ga.\cadr{\gb}{n}\ap \pi_1\aps \pi_m$ and $j > m$ } & \\
                         & \\
\text{undefined} \quad \text{otherwise} & 
\end{cases}
\end{array}
\]
\end{definition}

\begin{definition}\label{def:10.1.13-virtual}
Given a term $M$ we define a partial map $M_{(\cdot)} : \Seq \rightharpoonup \KTer{t}$ (extending that of Definition \ref{def:10.1.13}) as follows:
$$
M_\gs \ass
\begin{cases}
\pexp{M(\gt)}{\gt'}(\gt') & \text{if $\gt$ is the longest prefix of $\gs$ such that $\gt \in \dom{M}$, $\gs = \gt \cdot \gt'$ and} \\
                          & \text{$\pexp{M(\gt)}{\gt'}$ is defined} \\
                          &  \\
\text{undefined}          & \text{otherwise} \\
\end{cases}
$$
\end{definition}

The map $M_{(\cdot)}: \Seq \rightharpoonup \KTer{t}$ contains information about all possible $\eta$-expansions of the B\"{o}hm tree of $M$.
 We let $\virt{M} = \{\gs \in \Seq \sut M_\gs \text{ is defined } \}$. Note that $\dom{M} \subseteq \virt{M}$ and for every $\gs \in \dom{M}$, the values $M_\gs$ and $M(\gs)$ coincide. The elements of $\virt{M}$ are the
\emph{virtual sequences} of $M$. The map $M_{(\cdot)}$ extends $M(\cdot)$ by giving also \emph{virtual nodes}, which are intuitively nodes of some
 $\eta$-expansion of the B\"{o}hm tree of $M$, but still $M_{(\cdot)}$ cannot return the unreachable nodes, that correspond to sequences in
 $\Seq - \virt{M}$, which do not belong to any of the $\eta$-expansions of the B\"{o}hm tree of $M$. Note that for the maximal sequences
 $\gs \in \dom{M}$, the term $M_\gs$ may have an improper hnf or not have an hnf at all, while for non-maximal $\gs \in \dom{M}$, the term
 $M_\gs$ must have a proper hnf.


\begin{lemma}\label{lem:10.3.7bis}
Let $M,N$ be terms. If $\gs \in \virt{M}\cap\virt{N}$ is minimal such that $M_\gs \not\sim N_\gs$, then there exists a head context
 $C[\cdot]$ such that $C[M] \not\sim C[N]$. Moreover $C[M]$ has a proper hnf iff $M_\gs$ has a proper hnf and $C[N]$ has a proper hnf
 iff $N_\gs$ has a proper hnf.
\end{lemma}

\begin{proof}
If $\sigma$ is the empty sequence then the result is trivial. Now assume $\gs=\gs'(j,j')$, so that $M_{\gs'} \sim N_{\gs'}$
 and the similarity also holds for all prefixes of $\gs'$.
Let $\gt$ be the longest prefix of $\sigma'$ contained in $\dom{M}$ and let $\gt'$ be such that $\sigma'=\gt\gt'$.
 Let $\rho$ be the longest prefix of $\sigma'$ contained in $\dom{N}$ and let $\rho'$ be such that $\sigma'=\rho\rho'$.
 We assume w.l.o.g. that $\gt \geq \rho$. Let $n=\leng\sigma$, let $p$ be greater than $\bwei{M}{n}$, $\bwei{N}{n}$ and of all the
 second components of the pairs occurring in $\gs$. Let $q$ be greater than $\bbrea{M}{n}$, $\bbrea{N}{n}$ and of all the first component
 of the pairs occurring in $\gs$. Let $C[\cdot]$ be the context produced by Lemma \ref{lem:10.3.7} applied to $M$, $n$, $\rho$, $q$ and $p$.
 Then $C[M]$ and $C[N]$ reduce, respectively, to substitution instances $(M(\rho))^*$ and $(N(\rho))^*$ where the
 same substitutions have been applied. Therefore $(M(\rho))^*\sim (N(\rho))^*$ and $C[M]\sim C[N]$. Now let 
 $C'[\cdot]$ be the context produced by Lemma \ref{lem:10.3.7} applied to $\pexp{(N(\rho))^*}{\rho'}$, $n'=\leng\rho'$, $\gt'$, $q$ and $p$.
 Once again $C'[C[N]]$ reduces to a substitution instance $(\pexp{N(\rho)^*}{\rho'}(\rho'))^o$ and, because $M_{\gs'} \sim N_{\gs'}$,
 we have $C'[C[M]] \sim C'[C[N]]$. Finally let $C''[\cdot]$ be the context produced by Lemma \ref{lem:10.3.7} applied to
 $(\pexp{N(\rho)^*}{\rho'}(\rho'))^o$, $n''=1$, $\gt''=(j,j')$, $q$ and $p$. Then the context $C''[C'[C[\cdot]]]$ has the properties
 required in the statement, since $C''[C'[C[M]]] \not\sim C''[C'[C[N]]]$, $C''[C'[C[M]]]$ has a proper hnf iff $M_\gs$ has a proper hnf
 and $C''[C'[C[N]]]$ has a proper hnf iff $N_\gs$ has a proper hnf.
\end{proof}

\begin{theorem}[B\"{o}hm's theorem for the extended stack calculus]\label{cor:10.4.2}
Let $M,N$ be two distinct $\labelto{\s\eta}$-normal forms without subterms which are improper hnf's. Then $M$ and $N$ are separable.
\end{theorem}

\begin{proof}
Under the hypotheses of the statement about $M$ and $N$, there has to be a minimal sequence $\gs\in\virt{M}\cap\virt{N}$ such that $M_\gs \not\sim N_\gs$
 so that by Lemma \ref{lem:10.3.7bis} there exists a head context $C[\cdot]$ such that $C[M] \not\sim C[N]$.
 Now by hypothesis $M_\gs$ and $N_\gs$ are proper hnf's and hence $C[M]$ and $C[N]$ have proper hnf's. Therefore, applying
 Theorem \ref{thm:10.4.1} we have that $C[M]$ and $C[N]$ are separable, which trivially implies that $M$ and $N$ are separable.
\end{proof}

Supppose that there is a $\sigma \in \dom{M} \cap \dom{N}$ such that $M(\sigma) = \bd\seq\ga.\cadr{\nil}{h}\ap \seq\pi$ and
 $N(\sigma) = \bd\seq\ga'.\cadr{\nil}{h'}\ap \seq\varpi$. Assume $P$ is the $j'$-th term of the $j$-th stack of $\seq\pi$ and $Q$ is the $j'$-th term of the $j$-th stack of $\seq\varpi$ and that 
 $P\not\sim Q$. If $P$ and $Q$ are the only dissimilar subterms then no separating context can be built with the technique
 described in Lemma \ref{lem:10.3.7}. As a matter of fact the requirement, appearing in the statement of Theorem \ref{cor:10.4.2}, that $M$ and
 $N$ do not have subterms which are improper hnf's is more strict then necessary. In fact we only need that, among the dissimilar subterms, there
 is a $\gs\in\virt{M}\cap\virt{N}$ such that $M_\gs \not\sim N_\gs$ and $M_\gs,N_\gs$ are proper hnf's.\\
 We decided to study the extension of the stack calculus as defined in \cite{CaEhSa12}, thus including $\nil$ in the language.
 However $\nil$ received a special treatment throughout this paper, in the sense that the improper hnf's are kept do not play a role
 similar to the proper hnf's (see for example Definition \ref{def:op-equiv-extended}, Definition \ref{def:similarity} and Definition \ref{def:10.1.13}).
 The reason is that improper hnf's and terms without hnf are, in the extended stack calculus, in some sense comparable to what unsolvable terms are in the
 $\lambda$-calculus. In any case, we do not treat a notion of solvability for the extended stack calculus.

\subsection{Characterization of operational equivalence}

In this section we give a concrete characterization of operational equivalence that does not involve any universal quantification over head contexts
 (see Definition \ref{def:op-equiv}). Rather, this characterization is based on B\"{o}hm trees.


\begin{definition}\label{def:infinite-similarity}
We define a binary relation $\simty$ on $\KTer{t}$ as follows: $M \simty N$ iff $\virt{M}=\virt{N}$ and for all $\gs \in \virt{M}$
 we have that $M_\gs \sim N_\gs$ (up to rename of variables which are bound in some $M_\gt,N_\gt$, where $\gt < \gs$, but free in
 $M_\gs,N_\gs$).
\end{definition}


\begin{theorem}\label{thm:10.4.2}
$M \opeq N$ iff $M \simty N$.
\end{theorem}

\begin{proof}
\noindent ($\Rightarrow$) We prove the contrapositive. Suppose $M \not\simty N$. Then there exists a sequence $\gs$ such that $M \not\sim_\gs N$. Let 
 $\gs$ be minimal w.r.t. this property. Then $\gs \in \virt{M}\cap \virt{N}$, so that in fact $M_\gs \not\sim N_\gs$.

Since $M_\gs$ and $N_\gs$ are not similar, at least one among them must have a proper hnf: say it is $M_\gs$.
 Now suppose $N_\gs$ does not have a proper hnf. By Lemma \ref{lem:10.3.7bis} there exists
 a head context $C[\cdot]$ such that $C[M]$ has a proper hnf while $C[N]$ does not have a proper hnf. This proves that $M \not\opeq N$.\\
 If $N_\gs$ has a proper hnf, then by Lemma \ref{lem:10.3.7bis} there exists a head context $C'[\cdot]$ such that $C'[M] \not\sim C'[N]$ and both
 $C'[M]$, $C'[N]$ have a proper hnf; therefore applying Theorem \ref{thm:10.4.1} we obtain that $C'[M]$ and $C'[N]$ are separable and consequently
 by Theorem \ref{thm:sepa-nopeq}, we have $C'[M] \not\opeq C'[N]$. This trivially implies $M \not\opeq N$.\\

\noindent ($\Leftarrow$) Immediate, because by Theorem \ref{thm:HP-complete-extended}, the relation $\opeq$ is an HP-complete equational theory which, by 
 Theorem \ref{thm:10.4.2}, is contained in $\simty$.
 \end{proof}

%
%

\section{Conclusions}


The stack calculus \cite{CaEhSa12} is a finitary functional language in which the $\lambda\mu$-calculus can be faithfully translated, in the sense
 that conversion (and typing, for the typed versions) is preserved by the translation.
 As it happens for the $\lambda\mu$-calculus, the stack calculus fails to have the separation property and in this paper we introduce the
 extended stack calculus which, as Saurin's $\Lambda\mu$-calculus \cite{Saurin05}, does have this property. The separation property proved
 in this paper for the extended stack calculus has consequences both on the semantical and on the syntactical side. For example it implies
 that $=_{\s\eta}$ is the maximal consistent congruence on normalizable terms extending $=_\s$, so that any model of the extended stack
 calculus cannot identify two different $\s\eta$-normal forms without being trivial.

Nonetheless the definition of operational equivalence involves a universal quantification over contexts but the problem
 of checking operational equivalence between $\nil$-free normalizable terms reduces to the problem of finding their $\s\eta$-normal forms (Theorem \ref{cor:10.4.2})
 (with a leftmost strategy, for example). The complete characterization of operational equivalence (also for non-normalizable terms)
 is achieved: two terms of the extended stack calculus are operationally equivalent iff they have the same
 B\"{o}hm tree, up to possibly infinite $\eta$-expansion (Theorem \ref{thm:10.4.2}), a condition that does not involve a quantification
 over all head contexts. We showed that operational equivalence is maximally consistent, i.e. it cannot be properly extended to another consistent equational theory
 both in the stack calculus and in the extended stack calculus. We work out the details of a B\"{o}hm-out technique for the extended stack calculus  
 (Lemma \ref{lem:10.3.7}). A nice feature of the extended stack calculus is that, having only one binder, it admits a simpler proof of B\"{o}hm's
 theorem, which is similar to the one for the $\lambda$-calculus. Besides the applications of B\"{o}hm's theorem, there has always been interest around the proof itself and the algorithmic content of the B\"{o}hm-out
 technique: from Huet's \cite{Huet93} interest in the implementation and mechanical formalization of B\"{o}hm's proof, Aehlig and Joachimski
 \cite{Aehlig02} alternative proof and to Dezani et al.'s account \cite{DGP08}. Saurin in \cite{Saurin10} establishes a
 standardization theorem for the $\Lambda\mu$-calculus, and studies B\"{o}hm-like trees for the $\Lambda\mu$-calculus,
 strengthening the separation results that he obtained in \cite{Saurin05}. In view of the interest in proofs of B\"{o}hm's
 theorem for various calculi, we believe useful to contribute in the present work with a direct proof of B\"{o}hm's
 theorem (i.e. with a B\"{o}hm-out technique) for the extended stack calculus, even if the mere separation
 result would follow by a suitable mutual translation with the $\Lambda\mu$-calculus.
 
It is out of the scope of this paper to analyze the typed extended stack calculus. In the case of Saurin's $\Lambda\mu$
 the Curry\textendash Howard isomorphism carries through via the straightforward extension of the type system.
 However Saurin \cite{Saurin10b} and Nakazawa and Katsumata \cite{Naka12} noted that with this approach many ``interesting'' $\Lambda\mu$-terms
 are not typeable (for example those used in the B\"{o}hm out technique). For this reason Saurin \cite{Saurin10b} studies an alternative type system for the $\Lambda\mu$-calculus.
 This latter approach could also be adapted to the extended stack calculus.

\bibliographystyle{eptcs}
\bibliography{bibliography}


\end{document}